\documentclass[hidelinks,onefignum,onetabnum]{siamart220329}

\usepackage{lipsum}
\usepackage{amsmath,amssymb, amsfonts}
\usepackage{color}
\usepackage{graphicx}

\definecolor{myblue}{RGB}{94, 129, 181}
\definecolor{myorange}{RGB}{225, 156, 36}
\definecolor{mygreen}{RGB}{143, 176, 50}
\definecolor{myred}{RGB}{235, 98, 53}

\usepackage[T1]{fontenc}
\usepackage{charter}

\usepackage{mathtools}

\usepackage{setspace}
\usepackage{epstopdf}
\usepackage{algorithmic}
\ifpdf
  \DeclareGraphicsExtensions{.eps,.pdf,.png,.jpg}
\else
  \DeclareGraphicsExtensions{.eps}
\fi

\def\E{\mathbb{E}}

\def\bP{\mathbb P}

\def\R{\mathbb{R}}

\def\sA{{\mathcal A}}

\def\sF{{\mathcal F}}

\def\ub{{\underline{\beta}}}
\def\ob{{\overline{\beta}}}

\def\eps{\varepsilon}

\numberwithin{equation}{section}
\theoremstyle{plain}                % title and  number in bold, text italic

\newtheorem{assumption}[theorem]{Assumption}
\newtheorem{remark}{Remark}[section]

\title{Radner equilibrium with population growth\thanks{{\bf Funding:} Jin Hyuk Choi is supported by the National Research Foundation of Korea (NRF) grant funded by the Korea government (MSIT) (No. 002086221G0001278).}}

\author{Jin Hyuk Choi\thanks{Ulsan National Institute of Science and Technology, Ulsan, South Korea (\email{jchoi@unist.ac.kr}).}
\and Kim Weston\thanks{Department of Mathematics, Rutgers University, NJ, USA (\email{kw552@rutgers.edu}).}
}

\usepackage{amsopn}

\ifpdf
\hypersetup{
  pdftitle={A multi-agent targeted trading equilibrium with transaction costs},
  pdfauthor={J. Choi, J. Duraj, and K. Weston}
}
\fi

\begin{document}

\maketitle

\begin{abstract}
We prove the existence of a Radner equilibrium in a model with population growth and analyze the effects on asset prices.  A finite population of agents grows indefinitely at a Poisson rate, while receiving unspanned income and choosing between consumption and investing into an annuity with infinitely-lived exponential preferences.  After establishing the existence of an equilibrium for a truncated number of agents, we prove that an equilibrium exists for the model with unlimited population growth.  
%The proofs require an analysis of Lambert's $W$-function.  We present numerics to show that population growth affects the equilibrium annuity price by reducing its oscillatory amplitude relative to the constant population model.
Our numerics show that increasing the birth rate reduces oscillations in the equilibrium annuity price, and when younger agents prioritize the present more than older agents, the equilibrium annuity price rises compared to a uniform demographic.
\end{abstract}

% REQUIRED

% REQUIRED

\section{Introduction}\label{section:intro}

We study a continuous-time financial Radner equilibrium with a population of economic agents that can grow at a Poisson birth rate.  Traditionally, Radner equilibria have a fixed number of agents, whereas non-financial-economic notions of equilibrium often allow for population changes.\footnote{For example, models in mathematical biology and ecology study population dynamics, such as in Bolnick et. al. \cite{AABB11TEE} and Durrett and Levin~\cite{DL94RSL}.}  Whether there are a finite number of agents, such as in Duffie and Huang~\cite{DH85Econ}, a continuum of agents, such as in Mas-Colell and Vives~\cite{MV93RES}, or overlapping generations of agents, such as in Balasko and Shell~\cite{BS80JET}, the total mass of agents does not change over time.  In mathematical finance, incomplete Radner equilibria have a long history, including works by Cuoco and He~\cite{CH94wp}, Basak and Cuoco~\cite{BC98RFS}, Christensen et. al.~\cite{CLM12JET}, Choi and Larsen~\cite{CL15FS}, Kardaras et. al.~\cite{KXZ22}, and Escariaza et. al.~\cite{ESX22AAP}.  Our focus is on the previously unstudied aspect of incomplete Radner equilibrium that allows the population size to grow.

We introduce a model with a growing population in a Radner equilibrium.  Since market clearing must hold at all times, the market clearing equations shift after each agent birth, causing the price and agent strategies to shift too.  Thus, jumps in the population size cause jumps in the equilibrium-characterizing equation system, leading to a new form of equation system.

The agents have running consumption and income streams on an infinite time horizon with exponential preferences.  The financial market is incomplete and consists of a traded annuity, which pays a dividend stream of one per unit time.  The financial market is similar to that of Weston and \v{Z}itkovi\'{c}~\cite{WZ20FS}, who study a financial equilibrium with an annuity as the only traded security.  We make the assumption that the agents' income streams are arithmetic Brownian motions with drift, and the resulting equilibrium annuity price is a function of the number of agents in the economy.

Equilibrium is characterized by an infinite-dimensional recursion based on the number of agents in the economy.  The recursive system consists of one equation describing annuity prices and one equation for each of the infinitely many agents.  %As the population grows, the agents' equations do not converge to zero, as in the mean field limit of the stochastic game proposed by Tangpi and Wang~\cite{TW24FS}.  
In order to prove the existence of a solution to the equilibrium-characterizing recursive system, we study a truncated model that allows for only a finite number of agent births.  We prove the existence of a solution to the truncated model's recursive system.  Together with uniform bounds, we pass the truncated number of agent births to infinity to obtain a solution to the original recursive system.  The analysis of the truncated system involves a careful investigation of Lambert's $W$ function, which also appeared in Larsen and Sae Sue~\cite{LS16MAFE} when studying a Radner equilibrium with jumps in the agents' income streams.

The new feature of our model is the growing number of agents, and we analyze the impact that the birth rate has on prices.  Through analysis and numerics, we see that increases in the birth rate reduce the amplitude of oscillations in the annuity price.  The increase in population causes shifts in the population's demographics in terms of aggregate time preferences, income, and risk-aversion parameters.  %Numerical evidence suggests that 
When younger and older agents place different emphases on the present time compared to the future, the equilibrium annuity price changes compared to the homogeneous demographic case.  Through a numerical example, we exhibit that if younger agents place a larger emphasis on the present day compared to older agents, then annuity prices increase compared to the homogeneous case.

The paper is arranged as follows.  We set up the model in Section~\ref{section:setting}.  Section~\ref{section:main} presents the main results, Theorems~\ref{thm:existence-infinite} and~\ref{thm:eq-infinite}, which characterize and prove the existence of an equilibrium with a growing economy.  An equilibrium is constructed by analyzing a truncated model, proving uniform bounds, and passing to limiting case in Lemma~\ref{solutions1} and Theorem~\ref{thm:finite-verification}.  The proofs are contained in Section~\ref{section:proofs}.  Section~\ref{section:examples} provides numerical examples and investigates the impact of population growth in equilibrium.

\section{Model setting}\label{section:setting}

We study a continuous-time financial equilibrium in a pure exchange economy with an infinite time horizon.  We let $B=\{B_t\}_{t\geq 0}$ be a Brownian motion and $N=\{N_t\}_{t\geq 0}$ be a Poisson process with intensity $\lambda>0$ on the probability space $(\Omega,\sF,\bP)$.  $N_t$ represents the population size at time $t$.  The filtration $\{\sF_t\}_{t\geq 0}$ is generated by the processes $B$ and $N$.  For a random variable $X$ and stopping time $\tau$, we use the notation $\E_\tau[X] = \E[X|\sF_\tau]$.  For a stopping time $\tau$, $t\geq \tau$ refers to being on the event $\{t\geq \tau\}$.

\smallskip

\noindent{\bf The financial market.} 
The financial market consists of a traded annuity.  It is in one-net supply, and its price is denominated in units of a single consumption good.  The annuity pays a constant dividend stream of one per unit time, as in Weston and \v{Z}itkovi\'{c}~\cite{WZ20FS}.

The annuity price will be determined endogenously in equilibrium.  We focus on models where the annuity price, $A$, depends on the total number of agents, expressed as $A = \{A(N_t)\}_{t\geq 0}$. The annuity price process is an output of equilibrium.

\smallskip

\noindent{\bf Economic agents.} 
At time $t\geq 0$, the economy consists of $N_t<\infty$ economic agents who have the choice between consumption of the real good and investment in the financial market. Agents enter the economy at a $\text{Poisson}(\lambda)$ rate.  For $i \in \mathbb{N}$, we define
\begin{align*}
\tau_i:=\inf \left \{ t\geq 0 : \, N_t=i    \right \},
\end{align*}
and note that $\tau_1 = \ldots = \tau_{N_0} = 0$.

For agent $i \in \mathbb{N}$, positions held in the annuity over time are denoted $\theta = \{\theta_t\}_{t\geq\tau_i}$. Consumption occurs at the rate $c = \{c_t\}_{t\geq\tau_i}$ per unit time.  Agent $i$ receives the stochastic income stream $\eps_i = \{\eps_{i, t}\}_{t\geq \tau_i}$ with dynamics
$$
  d\eps_{i,t} = \mu_i(N_t) dt + \sigma_i(N_t) dB_t, \quad  \eps_{i,\tau_i}\in\R, \ \text{on } \{t\geq \tau_i\},
$$  
where $\mu_i$ and $\sigma_i$ are real-valued functions of $N_t$.  Agents $i=1,\ldots, N_0$ are endowed with $\theta_{i,0-}=\theta_{i,0}\in\R$ shares in the annuity.  Since the annuity is in one-net supply, we assume that $\sum_{i=1}^{N_0} \theta_{i,0} = 1$.  Agents $i>N_0$ are endowed with $\theta_{i,\tau_i-}=\theta_{i,\tau_i} = 0$ shares in the annuity and can begin consuming and investing after their birth at time $\tau_i$.
%have the ability to rebalance their holdings after they are born at time $\tau_i$.
\begin{definition}\label{admissibility}
  For a given agent $i \in \mathbb{N}$, the investment and consumption strategies $\theta$ and $c$ are {\bf admissible for agent $i$} if the following three conditions hold.
  \begin{enumerate}
    \item The processes $\theta$ and $c$ are progressively measurable with $\int_{\tau_i}^t |c_s|ds<\infty$ $\bP$-a.s. on $\{t\geq \tau_i\}$.  There exists a constant $M<\infty$ such that
    \begin{equation}\label{theta-growth}
      |\theta_t| \leq M(1+t), \quad \text{on}\quad \{t\geq \tau_i\}.
    \end{equation}
    
    \item On $\{t\geq\tau_i\}$, the associated wealth process $X_t=\theta_t A(N_t)$ satisfies
    \begin{equation}\label{self-financing}
      dX_t = \theta_{t-} (dA(N_t) + dt) + (\eps_{i,t}-c_t)dt.
    \end{equation}
    
    \item The transversality condition holds:
      \begin{align}\label{transversality}
        \lim_{t\rightarrow\infty} \E_{\tau_i}\left[A(N_t)\exp\left(-\rho_i(t-\tau_i)
       -\alpha_i \left(\theta_t+\eps_{i,t}\right)\right)\right]=0.
      \end{align}
  \end{enumerate}
  If $\theta$ and $c$ are admissible for agent $i$, then we write $(\theta, c)\in\sA_i$.

\end{definition}

The agents aim to maximize expected utility through continuous consumption over an infinite time horizon.  Each agent $i \in \mathbb{N}$ has a risk aversion parameter $\alpha_i>0$ and time preference parameter $\rho_i> 0$.  Agent $i$'s objective is
\begin{equation}\label{optimization}
  \sup_{(\theta, c)\in\sA_i} \E_{\tau_i}\left[\int_{\tau_i}^\infty -\exp\left(-\rho_i(t-\tau_i) - \alpha_i c_t\right) dt\right].
\end{equation}

We make the following standing assumption.

\begin{assumption}\label{asm:bdd} The agent income stream parameters are bounded. That is, for each agent $i \in \mathbb{N}$, we have $\{\mu_i(j)\}_{j\geq N_0\vee i}$ and $\{\sigma_i(j)\}_{j\geq N_0\vee i}$ are bounded.
\end{assumption}
% use \sigma_i's bounded in proof of Lemma \ref{lemma:mg}, proof of supermg prop of V, only bounds over all j's, not i's
% use both sigma_i and alpha_i bounds in proof of Lemma \ref{lemma:mg}, only bounds over all j's, not i's
% use mu_i(j) and sigma_i(j) bounds over all j's but not i's in proof of Thm \ref{thm:eq-infinite} when proving the (super)mg property for V

\smallskip

\noindent{\bf Equilibrium.} Our notion of equilibrium is a Radner equilibrium, but it differs from typical Radner equilibria due to the market clearing conditions adjusting with the birth of each new agent.
\begin{definition} Strategies $\{\theta_i,c_i\}_{i\in \mathbb{N}}$ and a price function $A = \{A(j)\}_{j\geq N_0}$ form an {\bf equilibrium} if
\begin{enumerate}
  \item {\it Strategies are optimal:} For each agent $i\in \mathbb{N}$, $(\theta_i, c_i)\in \mathcal{A}_i$ and solve \eqref{optimization}.
  
  \item {\it Markets clear:} For all $t\geq 0$, we have
  \begin{align*}
    \sum_{i=1}^{N_t} \theta_{i,t} &= 1, \quad    \sum_{i=1}^{N_t} c_{i,t} = 1 + \sum_{i=1}^{N_t} \eps_{i,t}.
  \end{align*}

\end{enumerate}

\end{definition}

\section{Equilibrium construction and existence}\label{section:main}
We construct an equilibrium using a solution $\left\{ A(j), k_i(j)\right\}_{j\geq N_0, 1\leq i\leq j}$ to the following system of recursive equations:
 \begin{align}
&   \frac{1}{A(j)} = \beta_{\Sigma}(j)+\lambda - \left( \alpha_{\Sigma}(j)    \sum_{i=1}^j \frac{ \lambda A(j+1) e^{-\alpha_i (k_i(j+1)-k_i(j))}}{\alpha_i}\right)\frac{1}{A(j)}, \label{Aj_original_infinite}\\
&k_i(j) =    \frac{(\beta_i(j)+\lambda) A(j)-1}{\alpha_i}  - \frac{ \lambda A(j+1) e^{-\alpha_i (k_i(j+1)-k_i(j))}}{\alpha_i}. \label{ki_original_infinite}
\end{align}
Here, the values $\{\alpha_{\Sigma}(j)\}_{j\geq N_0}$, $\{\beta_i(j)\}_{j\geq N_0, 1\leq i\leq j}$, and $\{\beta_{\Sigma}(j)\}_{j\geq N_0}$ are given by
\begin{align}
\begin{split}\label{param-defs}
  \alpha_{\Sigma}(j) &:= \left(\sum_{i=1}^{j} \frac{1}{\alpha_i}\right)^{-1} \quad \text{so that} \quad \frac{1}{\alpha_{\Sigma}(j)} = \sum_{i=1}^{j} \frac{1}{\alpha_i},\\
  \beta_i(j) &:= \rho_i +\alpha_i\mu_i(j) - \frac12\alpha_i^2 \sigma_i^2(j),\\
  \beta_{\Sigma}(j) &:= \alpha_\Sigma(j) \left(\sum_{i=1}^j \frac{\beta_i(j)}{\alpha_i}\right)
    \quad \text{so that} \quad \frac{\beta_\Sigma(j)}{\alpha_\Sigma(j)} = \sum_{i=1}^j \frac{\beta_i(j)}{\alpha_i}.
\end{split}
\end{align}
The value $\alpha_\Sigma(j)$ is the aggregate risk aversion when there are $j$ agents in the economy.  The value $\beta_i(j)$ is agent $i$'s time preference parameter adjusted for her income stream and risk aversion.  As we shall see below in Theorem~\ref{thm:finite-verification} (when applied with $N_0=j$ and $m=0$), the aggregate term $\beta_\Sigma(j)$ is the reciprocal of the (constant) price of the traded annuity in a non-growing economy with $j$ fixed agents.  In such a model with a fixed number $j$ of agents, $\beta_\Sigma(j)$ would represent the interest rate. The terms $(\beta_i(j)+\lambda)$ can be interpreted as agent $i$'s shadow risk premium for investing in the annuity.  In \eqref{ki_original_infinite}, $k_i(j)$ incorporates both a risk premium term and a hedging term that takes into account the annuity's price change from $j$ to $(j+1)$ total agents.\label{page-label-1}

We will impose the following assumption on the values $\{\beta_{i}(j)\}_{j\geq N_0,\, 1\leq i \leq j}$.

\begin{assumption}\label{asm:beta-bdd}
There exist constants $0<\ub\leq \ob $ such that $\ub\leq \beta_{i}(j)\leq \ob$ for $j\geq N_0$ and $1\leq i \leq j$. This assumption implies that $ \ub \leq  \beta_\Sigma(j)  \leq \ob$ for $j\geq N_0$.
\end{assumption}
%Theorem \ref{thm:existence-infinite} below supplies the needed existence result in order to follow through on Theorem \ref{thm:eq-infinite}'s equilibrium construction.  Together, Theorems \ref{thm:existence-infinite} and \ref{thm:eq-infinite} prove the existence of an equilibrium.
\begin{theorem}\label{thm:existence-infinite}
  Suppose that Assumptions \ref{asm:bdd} and \ref{asm:beta-bdd} hold. Then, there exists a solution  $\{A(j), k_i(j) \}_{j\geq N_0, 1\leq i \leq j}$ to \eqref{Aj_original_infinite}--\eqref{ki_original_infinite} that satisfies the following inequalities: 
\begin{align}
&\frac{1}{\ob + \lambda} < A(j) \leq \frac{1}{\ub}\quad \textrm{and}\quad A(j) e^{-\alpha_i k_i(j)} \leq \frac{1}{\ub}, \label{hypo-infinite} \\
&- \ln\left( \frac{\ob + \lambda}{\ub} \right) \leq \alpha_i k_i(j) \leq  \frac{\ob + \lambda}{\ub}-1. \label{ki_bounds}
\end{align}   
\end{theorem}
%Given a solution $\left\{ A(j), k_i(j)\right\}_{j\geq N_0, 1\leq i\leq j}$ to the equation system \eqref{Aj_original_infinite}--\eqref{ki_original_infinite}, we can construct an equilibrium.  
The following result, Theorem \ref{thm:eq-infinite} constructs an equilibrium, given a solution to \eqref{Aj_original_infinite}--\eqref{ki_original_infinite} supplied by Theorem~\ref{thm:existence-infinite}.

\begin{theorem}\label{thm:eq-infinite}
Suppose that Assumptions \ref{asm:bdd} and \ref{asm:beta-bdd} hold. By Theorem~\ref{thm:existence-infinite}, there exists a solution $\{A(j), k_i(j) \}_{j\geq N_0, 1\leq i \leq j}$ to \eqref{Aj_original_infinite}--\eqref{ki_original_infinite} satisfying \eqref{hypo-infinite}--\eqref{ki_bounds}. Then, $\{ A(N_t)\}_{t\geq 0}$ is an equilibrium annuity price corresponding to an equilibrium with optimal investment/consumption processes $\{ \theta_{i,t}, c_{i,t} \}_{t\geq \tau_i}$ defined by
\begin{equation}
\begin{split}\label{Xici}
%  X_{i,t} &:= A(N_t)\left(\theta_{i,\tau_i} - \int_{\tau_i}^t \frac{k_i(N_u)}{A(N_u)}du\right) \quad \textrm{on}\quad  \{t\geq \tau_i\}, \\
  \theta_{i,t} &:= \theta_{i,\tau_i} - \int_{\tau_i}^t \frac{k_i(N_u)}{A(N_u)}du \quad \textrm{on}\quad  \{t\geq \tau_i\}, \\
c_{i,t}&:=\theta_{i,t}+k_i(N_t)+\eps_{i,t} \quad \textrm{on}\quad \{t\geq \tau_i\}.
\end{split}
\end{equation}
%Moreover, we have
%\begin{align}
%\frac{1}{\ob + \lambda} < A(N_t) \leq \frac{1}{\ub} \quad \textrm{for}\quad t\geq 0.\label{bounds1-infinite}
%\end{align}
\end{theorem}

In \eqref{Xici}, $k_i(N_\cdot)$ represents both the rate of decrease (increase) per unit of annuity in the number of shares invested and the increase (decrease) in consumption. The proofs of Theorems \ref{thm:existence-infinite} and \ref{thm:eq-infinite} are presented in Section~\ref{section:proofs}.  
The system \eqref{Aj_original_infinite}--\eqref{ki_original_infinite} comprises an infinite number of recursive equations and lacks a natural initial or terminal condition. Instead of addressing the original system directly, we first investigate a truncated system with a finite number of equations and a terminal condition in Lemma~\ref{solutions1} and Theorem~\ref{thm:finite-verification}, where we interpret the finite number $m$ as the total number of agents born, or added, to the economy.
\begin{lemma}\label{solutions1}
Let $m\geq 0$. Consider the system of equations such that for $1\leq i \leq N_0+m$,
 \begin{align}
&A(N_0+m)= \frac{1}{\beta_{\Sigma}(N_0+m)}, \quad k_{i}(N_0+m) =     \frac{\beta_{i}(N_0+m) A(N_0+m) -1}{\alpha_i} , \label{Ak_terminal}
\end{align}
and \eqref{Aj_original_infinite}--\eqref{ki_original_infinite} for $N_0\leq j < N_0+m$ and $1\leq i \leq j$.
Under Assumptions \ref{asm:bdd} and \ref{asm:beta-bdd}, the system has a unique solution $\{A(j), k_i(j) \}_{N_0 \leq j \leq N_0+m, \, 1\leq i \leq j}$ such that $A(j)>0$.
The solution satisfies \eqref{hypo-infinite}--\eqref{ki_bounds} for $N_0\leq j \leq N_0+m$ and $1\leq i \leq j$.
\end{lemma}

%Next, we use the solution from Lemma~\ref{solutions1} to construct and prove the existence of an equilibrium for the case when the total number of agents added to the economy $m$ is finite.

\begin{theorem}\label{thm:finite-verification} Let $m\geq 0$, and suppose that $\{N_t\}_{t\geq 0}$ only jumps $m$ times so that
$$
  \tau_1 = \ldots = \tau_{N_0} = 0 < \tau_{N_0+1} < \ldots< \tau_{N_0+m} < \infty = \tau_{N_0+m+1} = \tau_{N_0+m+2} = \ldots
$$
Suppose that Assumptions \ref{asm:bdd} and \ref{asm:beta-bdd} hold. Let $\{A(j), k_i(j) \}_{N_0 \leq j \leq N_0+m, \, 1\leq i \leq j}$ be the unique solution given in Lemma~\ref{solutions1}.  Then, $\{ A(N_t)\}_{t\geq 0}$ is an equilibrium annuity price for an economy that grows by $m$ agents and the optimal  processes are as in \eqref{Xici} for $1\leq i\leq N_0+m$.
\end{theorem}
\begin{remark}[Non-constant birth rate $\lambda$]
  Our birth rate intensity is assumed to be a constant $\lambda>0$.  However, we could allow for $\lambda$ to be a function of the number of agents.  That is, the intensity of $N$ at time $t$ is given by $\lambda(N_t)$. We would need the boundedness of $\{ \lambda(j)\}_{j\geq N_0}$ for the same type of result. Lemma~\ref{solutions1} and Theorem~\ref{thm:finite-verification} can be seen as the special case that $\lambda(j)=\lambda$ for $N_0\leq j \leq N_0+m-1$ and $\lambda(j)=0$ for $j\geq N_0+m$. 
\end{remark}

\section{Proofs}\label{section:proofs}
%We begin with the proof of Lemma \ref{solutions1}, which establishes the existence of a unique solution to the equation system corresponding to an economy that only adds a finite number $m\geq 0$ of agents.
\begin{proof}[Proof of Lemma \ref{solutions1}]
Let $L:[-\frac{1}{e},\infty) \to [-1,\infty)$ be the inverse function of the map $w\mapsto w\, e^w$ ($L$ is called the {\it Lambert W function}). In other words, $w=L(z)$ if and only if $z=w\, e^w$. Since $L$ is strictly increasing on $[0,\infty)$,
\begin{align}
  &x>c>0 \quad \Longrightarrow \quad L(c \,e^x) < L(x \,e^x) = x, \label{f_1}
\end{align}
and observe that
\begin{align}
& x=c_1 + c_2 \, e^{c_3 x} \quad \Longrightarrow \quad x=c_1 - \tfrac{L(-c_2 c_3 \, e^{c_1 c_3})}{c_3}  \quad \textrm{for}\quad c_2\leq 0 \leq c_3.  \label{f_2}
\end{align}

We prove the lemma by backward induction on $j$. Since \eqref{ki_original_infinite} and \eqref{hypo-infinite} imply \eqref{ki_bounds}, we focus on the existence of the solution satisfying \eqref{hypo-infinite}.

\smallskip

\noindent {\bf Base case:} Let $j=N_0+m$. For $1\leq i \leq j$, \eqref{Ak_terminal} produces
\begin{align}
A(j) e^{-\alpha_i k_i(j)} = A(j) e^{-\beta_i(j) A(j) +1 } \leq  \tfrac{1}{\beta_i(j) },
\end{align}
where the inequality is due to $\max_{x\in \R} x\, e^{-b x +1} = \frac{1}{b}$ for $b>0$. Hence, for $j=N_0+m$, \eqref{hypo-infinite} holds.

\smallskip

\noindent {\bf Induction Hyphothesis:} Let $N_0\leq l<N_0+m$. Suppose that the system has a unique solution $\{A(j), k_i(j) \}_{l+1 \leq j \leq N_0+m,\, 1\leq i \leq j}$ that satisfies \eqref{hypo-infinite} for $l+1\leq j \leq N_0+m$. 

\smallskip

\noindent {\bf Induction Step:}  
Let $1\leq i\leq l$. Using \eqref{f_2}, we rewrite \eqref{ki_original_infinite} as
\begin{align}
 k_i(l) &=    \frac{(\beta_i(l)+\lambda) A(l) -1}{\alpha_i}  - \frac{L\left(\lambda A(l+1) e^{-\alpha_i k_i(l+1)+ (\beta_i(l)+\lambda)A(l)-1}\right)}{\alpha_i}. \label{ki_new}
\end{align}
Due to \eqref{ki_original_infinite} and \eqref{ki_new},
\begin{align}
\lambda A(l+1) e^{-\alpha_i (k_i(l+1)-k_i(j))} = L\left(\lambda A(l+1) e^{-\alpha_i k_i(l+1)+ (\beta_i(l)+\lambda)A(l)-1}\right). \label{ki_new2}
\end{align}
We apply \eqref{ki_new2} to rewrite \eqref{Aj_original_infinite} as
\begin{align}
    A(l) &= \tfrac{1}{\beta_{\Sigma}(l)+\lambda}  \Big(1+\alpha_{\Sigma}(l)    \sum_{i=1}^l \tfrac{L\left(  \lambda A(l+1) e^{-\alpha_i k_i(l+1)+ (\beta_i(l)+\lambda) A(l)-1}\right)}{\alpha_i}\Big). \label{Aj_new}
\end{align}
Our goal is to show that there exists $A(l)$ satisfying \eqref{Aj_new} and the inequalities in \eqref{hypo-infinite}. Let $f:[0,\infty) \to \R$ be defined as
\begin{align}
f(x):=\tfrac{1}{\beta_{\Sigma}(l)+\lambda} \Big(1+\alpha_{\Sigma}(l) \sum_{i=1}^l \tfrac{L\left(  \lambda A(l+1) e^{-\alpha_i k_i(l+1)+ (\beta_i(l)+\lambda) x -1}\right)}{\alpha_i}\Big) - x.
\end{align}
Direct computations produce
\begin{align}
f'(x) &= \tfrac{\alpha_{\Sigma}(l)}{\beta_{\Sigma}(l)+\lambda}  \sum_{i=1}^l \tfrac{\beta_i(l)+\lambda}{\alpha_i}\cdot \tfrac{L\left(  \lambda A(l+1) e^{-\alpha_i k_i(l+1)+ (\beta_i(l)+\lambda) x -1}\right)}{1+L\left(  \lambda A(l+1) e^{-\alpha_i k_i(l+1)+ (\beta_i(l)+\lambda) x -1}\right)} - 1 \nonumber\\
&<\tfrac{\alpha_{\Sigma}(l)}{\beta_{\Sigma}(l)+\lambda}  \sum_{i=1}^l \tfrac{\beta_i(l)+\lambda}{\alpha_i} - 1 =0. \label{f_derivative}
\end{align}
Due to \eqref{f_1}, if $(\beta_i(l)+\lambda) x -1> \lambda A(l+1) e^{-\alpha_i k_i(l+1)}$, then
\begin{align}
L\left(  \lambda A(l+1) e^{-\alpha_i k_i(l+1)+ (\beta_i(l)+\lambda) x -1}\right)<(\beta_i(l)+\lambda) x -1. \label{L_ineq1}
\end{align}
Therefore, for $x> \max_{1\leq i \leq l} \frac{\lambda A(l+1) e^{-\alpha_i k_i(l+1)} + 1}{\beta_i(l)+\lambda}$, \eqref{L_ineq1} implies
\begin{align}
f(x)<\tfrac{1}{\beta_{\Sigma}(l)+\lambda} \Big(1+\alpha_{\Sigma}(l) \sum_{i=1}^l \tfrac{(\beta_i(l)+\lambda) x -1}{\alpha_i}\Big) - x = 0. \label{f<0}
\end{align}
We also observe that $f$ is strictly concave and $f(0)>\frac{1}{\beta_{\Sigma}(l)+\lambda}$. Therefore, due to the inequalities in \eqref{f_derivative} and \eqref{f<0}, we conclude that there exists a unique solution of $f(x)=0$ for $x\in (0,\infty)$, and the solution, call it $A(l)$, satisfies
\begin{align}
\tfrac{1}{\beta_{\Sigma}(l)+\lambda}< A(l) \leq  \max_{1\leq i \leq l} \tfrac{\lambda A(l+1) e^{-\alpha_i k_i(l+1)} + 1}{\beta_i(l)+\lambda}.
\end{align}
The above inequalities and $A(l+1) e^{-\alpha_i k_i(l+1)} \leq \frac{1}{\ub}$ from the induction hypothesis imply $\frac{1}{\ob + \lambda} < A(l) \leq \frac{1}{\ub}$. We rewrite \eqref{ki_original_infinite} to obtain
\begin{align*}
A(l) e^{-\alpha_i k_i(l)} &= \tfrac{(1+\alpha_i k_i(l)) e^{-\alpha_i k_i(l)} +\lambda A(l+1) e^{-\alpha_i k_i(l+1)} }{\beta_i(l)+\lambda} \leq \frac{1}{\ub},
\end{align*}
where the inequality is due to $\max_{x\in \R} (1+x) e^{-x} =1$ and $A(l+1) e^{-\alpha_i k_i(l+1)} \leq \frac{1}{\ub}$ from the induction hypothesis. We conclude that the inequalities in \eqref{hypo-infinite} hold for $j=l$. 
\end{proof}

Using Lemma~\ref{solutions1}, we prove Theorem~\ref{thm:existence-infinite}, which shows the existence of a solution to the infinite recursive equation system \eqref{Aj_original_infinite}--\eqref{ki_original_infinite} satisfying \eqref{hypo-infinite}--\eqref{ki_bounds}. 

\begin{proof}[Proof of Theorem~\ref{thm:existence-infinite}]
The proof is based on a diagonalization argument.
For $m\in \mathbb{N}$, Lemma~\ref{solutions1} implies that there exists a solution $\{A^{(m)}(j), k_i^{(m)}(j) \}_{N_0 \leq j< N_0+m, \, 1\leq i \leq j}$ satisfying \eqref{Aj_original_infinite}, \eqref{ki_original_infinite}, \eqref{hypo-infinite} and \eqref{ki_bounds}.
Due to the uniform boundedness in \eqref{hypo-infinite} and \eqref{ki_bounds}, there is a sequence $\{ m_{N_0}(l) \}_{l\in \mathbb{N}} $ such that the following limits exist for $1\leq i\leq N_0$:
\begin{align*}
A(N_0)&:=\lim_{l\to \infty} A^{(m_{N_0}(l))}(N_0),\quad k_i(N_0):= \lim_{l\to \infty}  k_i^{(m_{N_0}(l))}(N_0).
\end{align*}
Again due to \eqref{hypo-infinite} and \eqref{ki_bounds}, there is a subsequence $\{m_{N_0+1}(l)\}_{l\in \mathbb{N}}\subset \{ m_{N_0}(l) \}_{l\in \mathbb{N}} $ such that the following limits exist for $1\leq i\leq N_0+1$:
\begin{align*}
A(N_0+1)&:=\lim_{l\to \infty} A^{(m_{N_0+1}(l))}(N_0+1),\quad k_i(N_0+1):= \lim_{l\to \infty}  k_i^{(m_{N_0+1}(l))}(N_0+1).
\end{align*}
We repeat the same argument to obtain a chain of subsequences 
\begin{align*}
\{ m_{N_0}(l) \}_{l\in \mathbb{N}} \supset \{ m_{N_0+1}(l) \}_{l\in \mathbb{N}} \supset \{ m_{N_0+2}(l) \}_{l\in \mathbb{N}} \supset \cdots
\end{align*}
such that for $j\geq N_0$ and $1\leq i \leq j$, the following limits exist:
\begin{align*}
A(j)&:=\lim_{l\to \infty} A^{(m_j(l))}(j),\quad k_i(j):= \lim_{l\to \infty}  k_i^{(m_j(l))}(j).
\end{align*}
We form the diagonal subsequence $\hat{m}(l):=m_l(l)$ for $l \geq N_0$. Then by construction, \begin{align}
A(j)=\lim_{l \to \infty} A^{(\hat{m}(l))}(j), \quad k_i(j)= \lim_{l\to \infty}  k_i^{(\hat{m}(l))}(j)\quad \textrm{for}\quad j\geq N_0, \,\, 1\leq i \leq j. \label{limit_def}
\end{align} 
Since $A^{(\hat{m}(l))}$ and $k_i^{(\hat{m}(l))}$ satisfy \eqref{Aj_original_infinite}--\eqref{ki_bounds}, the limits $A$ and $k_i$ in \eqref{limit_def} do as well.
\end{proof}

Before proving Theorem \ref{thm:eq-infinite}, we prove two helper results.
\begin{lemma}\label{lemma:mg1}
  Let $i \in \mathbb{N}$, $\Delta$ be predictable and $\E_{\tau_i}\left[\int_{\tau_i}^t |\Delta_s|ds\right] < \infty$ on $\{t\geq \tau_i\}$. Then,
\begin{align}
    \E_{\tau_i}\left[\int_{\tau_i}^t \Delta_s (dN_s-\lambda ds)\right] = 0, \quad \textrm{on} \quad  \{t\geq \tau_i\}.   \label{jump_integral_zero}
\end{align}
\end{lemma}
\begin{proof}
  By considering $\Delta = \Delta^+ - \Delta^-$, we can without loss of generality consider the case when $\Delta$ is nonnegative.  Since $\Delta$ is predictable, $\int_{\tau_i}^\cdot \Delta_s (dN_s-\lambda\, ds)$ is a local martingale on $[\tau_i, t]$.  Let $\{\eta_l\}_{l \in \mathbb{N}}$ be a localizing sequence.  For each $l \in \mathbb{N}$,
  \begin{align*}
    \E_{\tau_i}\left[\int_{\tau_i}^{t\wedge\eta_l} \Delta_s dN_s\right]
    = \E_{\tau_i}\left[\int_{\tau_i}^{t\wedge\eta_l} \Delta_s \lambda\,ds\right]  \quad \textrm{on} \quad  \{t\geq \tau_i\}.  
  \end{align*}
Since $\E_{\tau_i}\left[\int_{\tau_i}^t \Delta_s ds\right] < \infty$ on $\{t\geq \tau_i\}$, by the monotone convergence theorem, both sides of the above equation converge to a finite number as $l \to \infty$, and we conclude \eqref{jump_integral_zero}.
\end{proof}

\begin{lemma}\label{lemma:mg}
For $i \in \mathbb{N}$, let $(\theta, c)$ satisfy \eqref{theta-growth} and $\{W_t\}_{t\geq \tau_i}$ be defined as
\begin{align}
  \begin{split}\label{def:W}
  W_{t}&:=-A(N_t)\exp\left(-\rho_i (t-\tau_i)-\alpha_i\left(\theta_t+k_i(N_t) + \eps_{i,t}\right)\right).
  %  W_{i,t}&:=-A(N_t)\exp\left(-\rho_i (t-\tau_i)-\alpha_i\left(\frac{X_{i,t}}{A(N_t)}+k_i(N_t) + \eps_{i,t}\right)\right), \quad t\geq \tau_i,
  \end{split}
\end{align}
Then, $\E_{\tau_i}\left[\int_{\tau_i}^t W_{s}^2 ds\right]<\infty$ on $\{t\geq \tau_i\}$.
%we have
%\begin{align}
%&\E_{\tau_i} \left[\int_{\tau_i}^t - W_{i,s-} \alpha_i \sigma_i(N_{s-})dB_s \right. \nonumber\\
%&\qquad\qquad +  \left. W_{i,s-} \left(\frac{A(N_{s-}+1)}{A(N_{s-})}e^{-\alpha_i(k_i(N_{s-}+1)-k_i(N_{s-}))}-1\right)(dN_s -\lambda ds) \right] =0. \label{martingale}
%\end{align}
%\begin{align*}
 % \E_{\tau_i}\left[\int_{\tau_i}^t W_{s}^2 ds\right]<\infty   \quad \textrm{on} \quad \{ t\geq \tau_i\}.  
%  \quad \text{and} \quad
%  \E_{\tau_i}\left[\int_{\tau_i}^t W_{i,s}^2 ds\right]<\infty.
%\end{align*}
\end{lemma}
\begin{proof}
On $\{s\geq \tau_i\}$, we write $W_s^2$ in terms of the Doleans-Dade stochastic exponential $\mathcal{E}(\cdot)$,
\begin{align*}
  W_s^2 &= A(N_s)^2 e^{-2\alpha_i\left(\theta_t+k_i(N_s)+\eps_{i,\tau_i}\right) - 2\int_{\tau_i}^s\left( \rho_i + \alpha_i\mu_i(N_u) - \alpha_i^2\sigma_i^2(N_u)\right)du} \\
  % stopped writing calculation here
 &  \qquad \cdot\mathcal{E}\Big(-2\alpha_i\int_{\tau_i}^\cdot \sigma_i(N_u)dB_u\Big)_s\\
& \leq \tfrac{1}{\ub^2}e^{M'(1+s)-2\alpha_i \eps_{i,\tau_i} }\cdot \mathcal{E}\Big(-2\alpha_i\int_{\tau_i}^\cdot \sigma_i(N_u)dB_u\Big)_s
\end{align*}
for a constant $M'$ independent of $s$, where the inequality is due to \eqref{hypo-infinite}, \eqref{ki_bounds}, \eqref{theta-growth} and Assumption \ref{asm:bdd}.
Using the above inequality and $\eps_{i,\tau_i}\in \sF_{\tau_i}$, we conclude
\begin{align*}
  \E_{\tau_i}\left[\int_{\tau_i}^t W_s^2 ds\right]
  \leq \tfrac{1}{\ub^2}e^{-2\alpha_i \eps_{i,\tau_i}} \int_{\tau_i}^t e^{M' (1+s)}ds < \infty.
\end{align*}
\end{proof}

Finally, we prove Theorem~\ref{thm:eq-infinite}.  We omit the proof of Theorem~\ref{thm:finite-verification} because it is nearly identical to the proof of Theorem~\ref{thm:eq-infinite} with some minor modifications.
\begin{proof}[Proof of Theorem \ref{thm:eq-infinite}]
Let $\{A(j), k_i(j) \}_{j\geq N_0, 1\leq i \leq j}$ be a solution to \eqref{Aj_original_infinite}--\eqref{ki_original_infinite} satisfying \eqref{hypo-infinite}--\eqref{ki_bounds} by Theorem~\ref{thm:existence-infinite}. Recall that $\{ \theta_{i,t},c_{i,t} \}_{t\geq \tau_i}$ is defined in \eqref{Xici}.

\smallskip

\noindent{\bf Admissibility of $(\theta_i, c_i)$.} First, we show that $(\theta_i,c_i)\in \mathcal{A}_i$ for agent $i\in \mathbb{N}$. The expressions in \eqref{Xici} imply that $\theta_i$ and $c_i$ are progressively measurable. The inequalities in \eqref{hypo-infinite} and \eqref{ki_bounds} produce $\int_{\tau_i}^t |c_{i,s}|ds <\infty$ on $\{t\geq\tau_i\}$ and the bound \eqref{theta-growth} for $\theta_i$.
An application of It\^o's lemma using \eqref{Xici} produces \eqref{self-financing}: for $X_{i,t}:=\theta_{i,t} A(N_t)$,
 \begin{align*}
   dX_{i,t} &= -k_i(N_t)dt + \theta_{i,t-} dA(N_t)= \theta_{i,t-}\left(dA(N_t)+dt\right) + (\eps_{i,t}-c_{i,t})dt.
 \end{align*}

 It remains to check \eqref{transversality} for $\theta_i$. Let $\{W_{i,t}\}_{t\geq \tau_i}$ be as in \eqref{def:W} with $\theta$ replaced by $\theta_i$. Since $\{\alpha_ik_i(j)\}_{j\geq N_0,\, 1\leq i \leq j}$ is uniformly bounded by \eqref{ki_bounds}, to show \eqref{transversality}, it suffices to show that $ \lim_{t\rightarrow\infty} \E_{\tau_i}\left[ W_{i,t}\right] =0$. 
On $\{t\geq \tau_i\}$, It\^o's lemma produces
\begin{align}
  d W_{i,t} 
    &= W_{i,t}\left(\tfrac{\alpha_i k_i(N_t)}{A(N_t)}-\beta_i(N_t)\right)dt
     -  W_{i,t} \alpha_i \sigma_i(N_t) dB_t   \nonumber \\
    &\quad + W_{i,t-}\left(\tfrac{A(N_{t-}+1)}{A(N_{t-})}e^{-\alpha_i(k_i(N_{t-}+1)-k_i(N_{t-}))}-1\right) dN_t  \nonumber\\
    &=-\tfrac{W_{i,t}}{A(N_{t})} dt  - W_{i,t} \alpha_i \sigma_i(N_{t})dB_t  \nonumber \\
&\quad +  W_{i,t-} \left(\tfrac{A(N_{t-}+1)}{A(N_{t-})}e^{-\alpha_i(k_i(N_{t-}+1)-k_i(N_{t-}))}-1\right)(dN_t-\lambda dt), \label{dW}
\end{align}
 where the second equality is due to \eqref{ki_original_infinite}. We observe that \eqref{hypo-infinite} and \eqref{ki_bounds} imply
\begin{align}\label{Ak_bound}
0<\tfrac{A(j+1)e^{-\alpha_i k_i(j+1)}}{A(j)e^{-\alpha_i k_i(j)}} \leq \tfrac{\ob+\lambda}{\ub} \, e^{\frac{\ob + \lambda}{\ub}-1} \quad \textrm{for}\quad j\geq N_0, \,\, 1\leq i \leq j.
\end{align}
Assumption~\ref{asm:bdd} and \eqref{Ak_bound} enable us to apply Lemmas \ref{lemma:mg1} and \ref{lemma:mg} to \eqref{dW} and obtain
$$
\E_{\tau_i}\left[-W_{i,T+t} \right] = \E_{\tau_i}\left[-W_{i,T}\right] - \int_{T}^{T+t} \E_{\tau_i}\left[\frac{-W_{i,s}}{A(N_{s})} \right] ds \quad \text{for}\quad  t\geq 0  \,\,\text{ on } \,\, \{T\geq \tau_i\}.
$$
On $\{T\geq \tau_i\}$, we define $g(T):= \E_{\tau_i}\left[-W_{i,T}\right]$.  Then for $t\geq 0$ on $\{T\geq \tau_i\}$, \eqref{hypo-infinite} implies that
\begin{align}\label{bound1}
  g(T+t) \leq g(T) -  \ub  \int_T^{T+t} g(s) ds.
\end{align}
The nonnegativity of $g$ and \eqref{bound1} imply that $g$ is decreasing and $\int_{\tau_i}^\infty g(s)ds <\infty$.  Therefore, $\lim_{t\rightarrow\infty} g(t) = 0$, and $ \lim_{t\rightarrow\infty} \E_{\tau_i}\left[ W_{i,t}\right] =0$ is proven.
 
\smallskip

\noindent{\bf Optimality.} For agent $i\in \mathbb{N}$, let $(\theta,c)\in \mathcal{A}_i$ and $W_t$ be defined as in \eqref{def:W}. We define $\{V_t\}_{t\geq \tau_i}$ as
$$
  V_{t}:=W_t - \int_{\tau_i}^t e^{-\rho_i(u-\tau_i) - \alpha_i c_{u}} du \quad \textrm{on} \quad \{t\geq \tau_i\},
$$
Let $\{V_{i,t}\}_{t\geq \tau_i}$ be as above in terms of $\{ \theta_{i,t}, c_{i,t}\}_{t\geq \tau_i}$.
It\^o's lemma produces
\begin{align}
&dV_t = \mu_{V,t} dt + \sigma_{V,t}dB_t + \phi_{V,t-}(dN_t -\lambda dt),\nonumber\\
&\textrm{with}  \begin{cases}
  \mu_{V,t} =-e^{-\rho_i(t-\tau_i)-\alpha_ic_t} -W_t \Big(\beta_i + \frac{\alpha_i}{A(N_t)}\big(\theta_t+\eps_{i,t}-c_t\big) \\
  \quad\qquad\qquad\qquad\qquad\qquad -\lambda\big(\frac{A(N_t+1)}{A(N_t)}e^{-\alpha_i(k_i(N_t+1)-k_i(N_t))}-1\big)\Big),  \\
\sigma_{V,t} =- \alpha_i\sigma_i(N_t)W_t, \\
   \phi_{V,t} = W_t\left(\frac{A(N_t+1)}{A(N_t)}e^{-\alpha_i(k_i(N_t+1)-k_i(N_t))}-1\right). % \label{sigma_V}
\end{cases}
\label{mu_V}
\end{align}
Assumption~\ref{asm:bdd} and \eqref{Ak_bound} enable us to apply Lemmas \ref{lemma:mg1} and \ref{lemma:mg} and obtain
\begin{align}\label{martingale2}
  \E_{\tau_i}\left[\int_{\tau_i}^t \sigma_{V,s} dB_s + \int_{\tau_i}^t \phi_{V,s-} (dN_s-ds)\right]=0 \quad \text{on}\quad \{t\geq\tau_i\}.
\end{align}

Fenchel's inequality tells us that for all $x\in\R$, we have
$$
  -e^{-\alpha_ix}\leq y\left(\ln y - 1\right)+\alpha_i xy \quad \text{for}\quad y>0.
$$
Taking $x=c_t$ and $y= -e^{\rho_i(t-\tau_i)} W_t / A(N_t)$ gives us
\begin{align}\label{fenchel}
  -e^{-\alpha_ic_t}\leq  -e^{\rho_i(t-\tau_i)} \tfrac{W_t}{A(N_t)} \big( -\alpha_i\left(\theta_t + k_i(N_t)+\epsilon_{i,t} \right) - 1 + \alpha_i c_t \big),
\end{align}
where the inequality becomes an equality if $c_t = c_{i,t}$. Combining \eqref{mu_V} with \eqref{fenchel} and \eqref{ki_original_infinite} yields
\begin{align}
 \mu_{V,t} 
 \leq 
% \frac{W_t}{A(N_t)} \big(\alpha_i k_i(N_t)-A(N_t)(\beta_i(N_t)+\lambda)+1+\lambda A(N_t+1)e^{-\alpha_i \left(k_i(N_t+1)-k_i(N_t)\right)}\big)%=0, 
 \tfrac{W_t}{A(j)}   \Big( \alpha_i k_i(j) -    (\beta_i(j)+\lambda) A(j)+1  + \lambda A(j+1) e^{-\alpha_i (k_i(j+1)-k_i(j))}\Big)\Big|_{j=N_t} = 0,   \label{mu_V_ineq}
\end{align}
where the inequality becomes an equality if $c_t = c_{i,t}$. By \eqref{mu_V_ineq} and \eqref{martingale2}, we obtain
\begin{align}
\E_{\tau_i}[V_{t}] \leq \E_{\tau_i}[V_{i,t}] \quad \text{on}\quad \{t\geq\tau_i\}. \label{V<V_i}
\end{align}
For any $(\theta, c)\in\sA_i$, the monotone convergence theorem, \eqref{transversality} and  \eqref{V<V_i} produce
\begin{align*}
  \E_{\tau_i}\left[\int_{\tau_i}^\infty -e^{-\rho_i(s-\tau_i)-\alpha_i c_s}ds\right]  
  %&= \lim_{t\rightarrow\infty}\E_{\tau_i}\left[\int_{\tau_i}^t -e^{-\rho_i(s-\tau_i)-\alpha_i c_s}ds\right]
  &= \lim_{t\rightarrow\infty} \E_{\tau_i}[V_t] \\
  &\leq \lim_{t\rightarrow\infty} \E_{\tau_i}[V_{i,t}] 
  %&= \lim_{t\rightarrow\infty} \E_{\tau_i}\left[\int_{\tau_i}^t -e^{-\rho_i(s-\tau_i)-\alpha_i c_{i,s}}ds\right]\\
  = \E_{\tau_i}\left[\int_{\tau_i}^\infty -e^{-\rho_i(s-\tau_i)-\alpha_i c_{i,s}}ds\right].
\end{align*}
Therefore, $(\theta_i,c_i)$ is optimal for agent $i$'s utility maximization problem \eqref{optimization}.

\smallskip

\noindent{\bf Market clearing holds.} We need to prove that market clearing holds in both the annuity market and the real goods market. By \eqref{Aj_original_infinite} and \eqref{ki_original_infinite},
\begin{align}
  \sum_{i=1}^j k_i(j) 
  &= \sum_{i=1}^j \tfrac{(\beta_i(j)+\lambda) A(j)-1}{\alpha_i}  - \sum_{i=1}^j\tfrac{ \lambda A(j+1) e^{-\alpha_i (k_i(j+1)-k_i(j))}}{\alpha_i}
%  &= \frac{(\beta_\Sigma(j)+\lambda)A(j)-1}{\alpha_\Sigma(j)} - \left(\beta_\Sigma(j)+\lambda-\frac{1}{A(j)}\right)\frac{A(j)}{\alpha_\Sigma(j)}
= 0. \label{ki_sum=0}
\end{align}
For $i \in \mathbb{N}$ and $t\geq 0$, the definition of $\theta_i$ in \eqref{Xici} and \eqref{ki_sum=0} imply that 
\begin{align}
  \sum_{i=1}^{N_t} \theta_{i,t}  
  &= \sum_{i=1}^{N_t} \Big( \theta_{i,\tau_i} - \sum_{j=i}^{N_t} \int_{\tau_j\wedge t}^{\tau_{j+1}\wedge t} \tfrac{k_i(j)}{A(j)} du   \Big) \nonumber\\
 & = 1 - \sum_{j=1}^{N_t}  \int_{\tau_j\wedge t}^{\tau_{j+1}\wedge t} \frac{1}{A(j)}\Big(\sum_{i=1}^j k_i(j) \Big) du=1. \label{annuity_clear}
\end{align}
Hence, market clearing holds for the annuity market.
For the real goods market clearing, \eqref{ki_sum=0} and \eqref{annuity_clear} produce the desired result:
\begin{align*}
  \sum_{i=1}^{N_t} c_{i,t}
  &= \sum_{i=1}^{N_t} \theta_{i,t}  + \sum_{i=1}^{N_t} \eps_{i,t} + \sum_{i=1}^{N_t} k_i(N_t) = 1 + \sum_{i=1}^{N_t} \eps_{i,t}.
\end{align*}
\end{proof}

\section{Numerical Examples}\label{section:examples}
First we recall the definitions of $\{\beta_i(j)\}_{j\geq N_0, 1\leq i\leq j}$ and $\{\beta_{\Sigma}(j)\}_{j\geq N_0}$ from \eqref{param-defs}:
$$
  \beta_i(j) := \rho_i +\alpha_i\mu_i(j) - \frac12\alpha_i^2 \sigma_i^2(j)
  \quad \text{and} \quad
  \beta_{\Sigma}(j) := \alpha_\Sigma(j) \Big(\sum_{i=1}^j \tfrac{\beta_i(j)}{\alpha_i}\Big).
$$
All agent parameters ($\rho_i$, $\alpha_i$, $\mu_i(j)$, $\sigma_i(j)$) contribute to $\beta_i(j)$ values, which are aggregated to produce $\beta_\Sigma(j)$.  Moreover, the $\beta_i(j)$ values are fundamental to determining $\alpha_i k_i(j)$ and the annuity price, whereas all other agent input parameters affect those equilibrium outcomes only through $\beta_i(j)$.  For this reason, we focus our numerical examples on $\beta_i(j)$- and $\lambda$-dependence without considering specific variations in the other agent input parameters.

\begin{figure}[t]
  \begin{center}
  \includegraphics[width=2.5in]{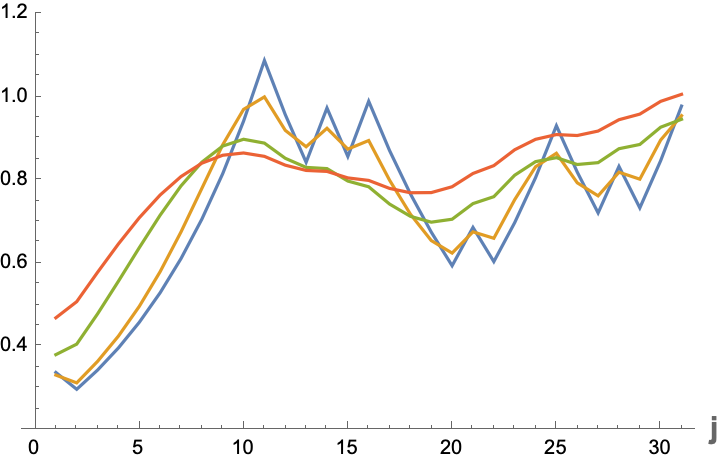}
  \end{center}
 %   \vspace{-0.2in}
  \caption{We assign arbitrary values to the process $\{\beta(N_0+j)\}_{1\leq j \leq 50}$ for \eqref{beta}. We plot $\{1/\beta(N_0+j)\}_{1\leq j \leq 30}$ (\textcolor{myblue}{blue} line) and the equilibrium annuity price processes $\{A(N_0+j)\}_{1\leq j \leq 30}$ for varying $\lambda$ (\textcolor{myorange}{orange}, \textcolor{mygreen}{green}, \textcolor{myred}{red} lines for $\lambda=1, 5,10$) are plotted. We set $\alpha_i=0.5$ for all $i$. 
}
\label{fig1}
\end{figure}

We first illustrate the dependency of the equilibrium annuity price process on $\lambda$ in Figure~\ref{fig1}. Let $m$ represent the total number of agents born, as stated in Lemma~\ref{solutions1} and Theorem~\ref{thm:finite-verification}. For simplicity, we assume that 
\begin{align}
\beta_i(j)=\beta(j) \quad \textrm{for} \quad N_0 \leq j \leq N_0+m, \,\, 1\leq i \leq j. \label{beta}
\end{align}
Then, we have $\beta_\Sigma(j)=\beta(j)$ and inductively check that the system in Lemma~\ref{solutions1} produces $k_i(j)=0$ for $N_0\leq j\leq N_0+m$ and $1\leq i \leq j$. As a result, the system is simplified as
\begin{equation}
\begin{split}\label{simple_A}
A(N_0+m)&= \tfrac{1}{\beta(N_0+m)}, \quad A(j)=\tfrac{\lambda A(j+1) + 1}{\lambda + \beta(j)}\quad \textrm{for} \ \ N_0 \leq j < N_0+m.
\end{split}
\end{equation}
Observe that $A(j)=\frac{1}{\beta(j)}$ when $\lambda=0$. This produces
\begin{align*}
\tfrac{\partial}{\partial \lambda} A(j) \Big|_{\lambda=0} 
&= \tfrac{(A(j+1) + \lambda \frac{\partial}{\partial \lambda} A(j+1))(\lambda + \beta(j))- (\lambda A(j+1)+1)}{(\lambda + \beta(j))^2}\Big|_{\lambda=0}\\
&=\tfrac{1}{\beta(j)} \left(\tfrac{1}{\beta(j+1)} - \tfrac{1}{\beta(j)} \right).
\end{align*}
Therefore, for $\lambda$ close to zero and $\frac{1}{\beta(j+1)}\neq \frac{1}{\beta(j)}$, $A(j)$ should be located between $\frac{1}{\beta(j+1)}$ and $\frac{1}{\beta(j)}$. This observation aligns with the graphs in Figure~\ref{fig1}, which show that the equilibrium annuity process $A$ exhibits a reduced amplitude of oscillation (tracking $1/\beta$ more loosely) as $\lambda$ increases.

\begin{figure}[t]
  \begin{center}
  \includegraphics[width=2.5in]{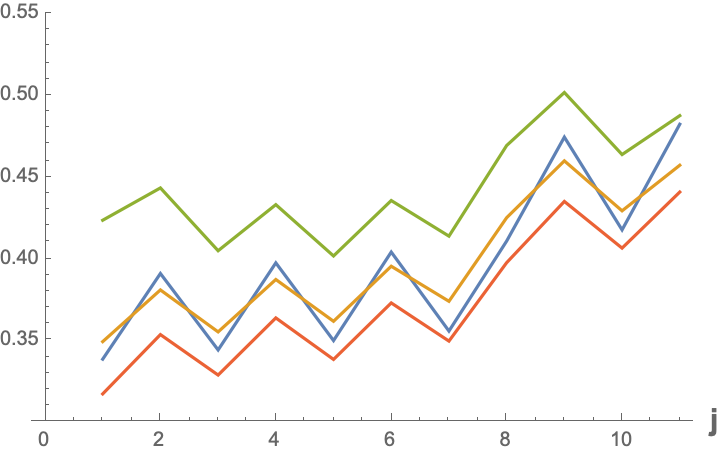}
  \end{center}
%    \vspace{-0.2in}
  \caption{We assign arbitrary values to the process $\{\beta_{\Sigma}(N_0+j)\}_{1\leq j \leq 50}$ and generate $\beta_i(j)$ parameters by \eqref{hetero} with $a=5$. We plot $\{1/\beta_{\Sigma}(N_0+j)\}_{1\leq j \leq 10}$ (\textcolor{myblue}{blue} line) and the equilibrium annuity price processes $\{A(N_0+j)\}_{1\leq j \leq 10}$ for different scenarios in \eqref{hetero} are plotted. We set $\lambda=1$ and $\alpha_i=0.5$ for all $i$. }
\label{fig2}
\end{figure}

Second, we demonstrate how the heterogeneity of $\{\beta_i(j)\}_{1\leq i \leq j}$ among agents can impact the equilibrium annuity price process. In Figure~\ref{fig2}, given the process $\{\beta_{\Sigma}(N_0+j)\}_{1\leq j \leq 50}$, we plot  $\{ 1/{\beta_{\Sigma}(N_0+j)}\}_{1\leq j \leq 10}$ as the \textcolor{myblue}{blue} line. The \textcolor{myorange}{orange}, \textcolor{mygreen}{green}, and \textcolor{myred}{red} lines are graphs of $\{A(N_0+j)\}_{1\leq j \leq 10}$, generated using the varying $\beta_i(j)$ parameters given by
\begin{align}
\beta_i(j) =
\begin{cases}
 \beta_\Sigma(j), &\textrm{for the \textcolor{myorange}{orange} line}\\
   \beta_\Sigma(j) - a + 2a\cdot \frac{i-1}{j-1}, &\textrm{for the \textcolor{mygreen}{green} line}\\
 \beta_\Sigma(j) - a + 2a\cdot \frac{j-i}{j-1}, &\textrm{for the \textcolor{myred}{red} line} 
\end{cases},\label{hetero}
\end{align}
for a constant $a>0$. The $\beta_i(j)$ parameters are chosen so that $\beta_\Sigma(j)$ remains the same amongst the \textcolor{myorange}{orange}, \textcolor{mygreen}{green}, and \textcolor{myred}{red} lines for all $1\leq j\leq 50$. The \textcolor{myorange}{orange} line represents the annuity price in the homogeneous scenario, and it exhibits the same reduced oscillatory behavior of $\{1/\beta_\Sigma(N_0+j)\}_{1\leq j\leq 10}$ seen in Figure~\ref{fig1}.

In Figure~\ref{fig2}, the \textcolor{mygreen}{green} scenario introduces heterogeneity into the choice of $\beta_i(j)$ amongst agents $i$ for each fixed j.  Younger agents, corresponding to larger $i$ values, have larger $\beta_i(j)$ values.  Hence, the younger agents place more importance on the present day than the future in \eqref{optimization}.  All \textcolor{mygreen}{green} scenario agents decrease their relative time preference $\beta_i(j)- \beta_{\Sigma}(j)$ as they get older, where getting older corresponds to larger $j$ values.  As all agents, young and old, age, they increasingly care more about the future than the present in \eqref{optimization}.  The $\beta_i(j)$ heterogeneity exhibited in the \textcolor{mygreen}{green} scenario leads to an increase in annuity prices compared to the \textcolor{myorange}{orange} homogeneous scenario.  Conversely, the annuity prices in the heterogeneous \textcolor{myred}{red} scenario are decreased due to younger agents placing more importance on the future, relative to the older agents.

% with homogenous agents, while the \textcolor{mygreen}{green} and \textcolor{myred}{red} lines represent the heterogenous case: the \textcolor{mygreen}{green} (\textcolor{myred}{red}, respectively) line corresponds to the scenario where younger agents have {\color{blue}higher (lower, respectively)} $\beta_i(j)$ values. Recall that $\beta_i(j)$ represents agent $i$'s time preference parameter, adjusted for her income stream and risk aversion. 
%Figure~\ref{fig2} suggests that {\color{blue}when younger generations have higher $\beta_i(j)$ values, they place more emphasis on the present time and less on the future, which leads to an increase in the annuity price over the homogeneous case.}

\begin{figure}[t]
  \begin{center}
$  \begin{array}{cc}
  \includegraphics[width=2in]{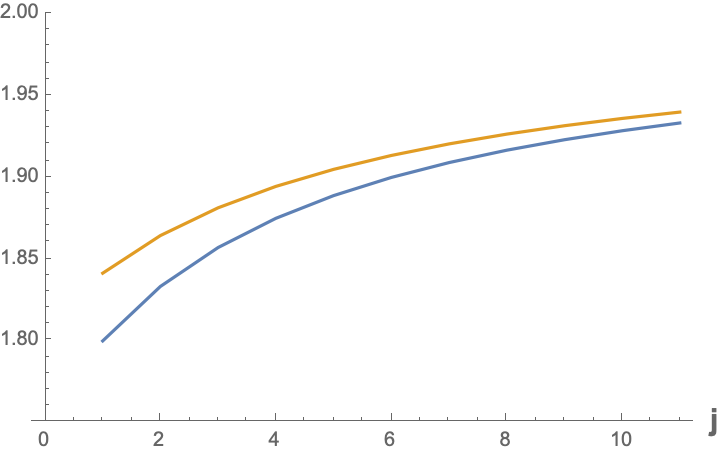} &   \includegraphics[width=2in]{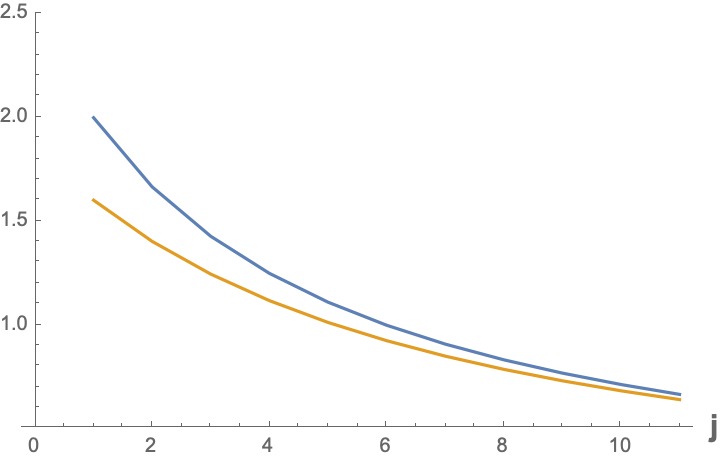} 
%\text{Decreasing $\beta_{\Sigma}$} &\text{Increasing $\beta_{\Sigma}$}
\end{array}$
  \end{center}
%  \vspace{-0.1in}
  \caption{We consider homogenous agents with $\beta_{\Sigma}(j)=\frac{j}{2j-1}$ for the left and $\beta_{\Sigma}(j)=0.1j$ for the right, $N_0\leq j \leq N_0+50$.  
  In both graphs, the \textcolor{myblue}{blue} line is the graph of $\{1/\beta_{\Sigma}(N_0+j)\}_{1\leq j \leq 10}$ and the \textcolor{myorange}{orange} line is the graph of the corresponding equilibrium annuity price processes $\{A(N_0+j)\}_{1\leq j \leq 10}$, computed by \eqref{simple_A} with $N_0=5, m=50, \lambda=1$, $\alpha_i=0.5$ for all $i$.    
}
\label{fig3}
\end{figure}

Lastly, Figure~\ref{fig3} depicts the equilibrium annuity price process for increasing or decreasing $\beta_{\Sigma}(j)$ with respect to $j$. We observe that if $\beta_{\Sigma}(j)$ decreases (or increases, respectively) in $j$, the annuity price $A(j)$ is higher (or lower, respectively) than $1/\beta_{\Sigma}(j)$.

\bibliographystyle{siamplain}
\bibliography{growing-economy-bib}

\end{document}